\documentclass[conference]{IEEEtran}
\IEEEoverridecommandlockouts

\usepackage{amsmath, bbm}
\usepackage{amssymb}
\usepackage{amsthm}
\usepackage{xcolor}

\theoremstyle{definition} 

\theoremstyle{theorem}

\newtheorem{theorem}{Theorem}
\newtheorem{lemma}{Lemma}
\newtheorem{corollary}{Corollary}
\usepackage{caption}
\usepackage{booktabs}

\usepackage{tabularx}

\usepackage{mathtools}

\newcommand{\myitem}[1]{\vspace{3pt}\noindent\textbf{#1}}

\usepackage{graphicx}
\usepackage[colorlinks=true, allcolors=blue]{hyperref}

\usepackage[noend]{algorithmic}
\usepackage[linesnumbered,ruled,vlined,algo2e]{algorithm2e} %
\usepackage{algorithm}
\let\oldnl\nl
\newcommand{\nonl}{\renewcommand{\nl}{\let\nl\oldnl}}

\usepackage{color}
\definecolor{red}{rgb}{1,0.2,0.2}
\definecolor{green}{rgb}{0.2,1,0.5}
\definecolor{blue}{rgb}{0,0,1}
\definecolor{lightblue}{rgb}{0.3,0.5,1}

\newcommand{\Prob}[1]{\mathrm{Pr}(#1)}

\newcommand{\ket}[1]{|#1\rangle}
\newcommand{\bra}[1]{\langle#1|}

\renewcommand{\baselinestretch}{0.99}

\title{On the Swapping Capacity of a Quantum Repeater
\thanks{Any mention of commercial products in this paper is for information only; it does not imply any recommendation or endorsement by NIST. U.S. Government work not protected by U.S. copyright.}
}

\author{
\IEEEauthorblockN{
Van Sy Mai\textsuperscript{1},
Richard J. La\textsuperscript{2},
Abdella Battou\textsuperscript{1}, 
Abderrahim Amlou\textsuperscript{1,3},
Cory Nunn\textsuperscript{4}
}

\IEEEauthorblockA{
\textsuperscript{1}CTL, NIST, Gaithersburg, MD, USA,\quad
\textsuperscript{2}ECE, Univ of Maryland, College Park, MD, USA\\
\textsuperscript{3}MSTII, Univ of Grenoble Alpes, Grenoble, France,\quad \textsuperscript{4}PML, NIST, Gaithersburg, MD, USA,
}

\IEEEauthorblockA{
vansy.mai@nist.gov,
hyongla@umd.edu,
abdella.battou@nist.gov, 
abderrahim.amlou@nist.gov,
cory.nunn@nist.gov
}
}

\begin{document}
\maketitle

\begin{abstract}
We study the capacity of a memory-based quantum repeater in entanglement swapping between two quantum links with either single or multiple memories, which we refer to as the end-to-end (E2E) entanglement throughput, subject to a constraint on the minimum fidelity. In order to approximate the E2E entanglement throughput, we adopt queueing models, where quantum links can have different characteristics: memory capacities, entanglement attempt rates and success probabilities, as well as classical communication latencies. We develop a model for estimating  E2E entanglement fidelity, while taking into account the heterogeneous dephasing and depolarizing dynamics of quantum memories and Bell-state measurements in entanglement swapping as well as classical communication delays and noises. Finally, with the help of our models for approximating the E2E entanglement throughput and fidelity, we use the maximum waiting times of entanglements in quantum memories at the repeater as optimization variables to maximize the E2E entanglement throughput while ensuring required minimum E2E fidelity. 
\end{abstract}

\section{Introduction}
Quantum networks are envisioned to enable non-classical applications, such as blind quantum computing, quantum sensors, quantum key distribution, distributed quantum computation, and ultimately, the quantum internet \cite{kimble2008quantum, wehner2018quantum}. A core requirement for these applications is end-to-end (E2E) entanglement distribution, which relies on quantum memory networks to establish high-quality, long-distance quantum connections. However, establishing and maintaining entanglements over long distances is challenging because they are highly susceptible to decay and decoherence from the environment and interfaces. Furthermore, in quantum memory networks, the quality of memory suitable for communication is currently very limited in both efficiency and coherence time, which restricts both the range and the quality of the achievable E2E entanglement. 

In many cases, applications that consume the provided E2E entanglements have a requirement on their fidelity. We consider a simple quantum network consisting of two quantum links and a quantum repeater between them which performs entanglement swapping. Our goal is to devise a framework that can be used to maximize the rate at which E2E entanglements can be provided to the end nodes subject to a constraint on their fidelity.

\noindent \underline{\bf Motivation for Our Study:} 
For a general path, it has been shown that the E2E entanglement rate depends on the swapping order and we can find an optimal swapping order efficiently, which maximizes the throughput~\cite{mai2025towards, chang2022order, ghaderibaneh2022efficient,shchukin2019waiting}. For a heterogeneous path with probabilistic swapping, an optimal policy can be a fixed order that can be described by a binary tree (i.e., each node has exactly two children). Thus, the network with two quantum links considered here serves as a building block for analyzing the rate of the tree; this result is omitted here due to a space constraint and will be reported elsewhere. Moreover, it also finds application in analyzing the performance of a quantum switch in a general network. 

\noindent \underline{\bf Summary of Our Contributions:}
First, we develop a model for estimating the E2E entanglement fidelity, which takes into account heterogeneous dephasing and depolarizing dynamics of quantum memories, Bell-state measurements (BSMs) in entanglement swapping, and classical communication delays and noises. We show that a constraint on minimum E2E entanglement fidelity is equivalent to a constraint on the maximum holding times (MHTs) of entangled qubits at the repeater before they are discarded. 

Second, based on this observation, we propose analytical models to approximate the E2E entanglement throughput subject to a constraint on MHTs of the stored entangled qubits at the repeater.
They provide formulae that can be used to approximate the E2E entanglement throughput as a function of the MHTs at the repeater. 

Third, using the above models,
we formulate the problem of maximizing the E2E entanglement throughput subject to a constraint on E2E entanglement fidelity as an optimization problem with the MHTs as the optimization variables. 
We show that the optimal solution of an equivalent problem is a corner point of the feasible set, in the process revealing the insights behind the solution.

Finally, we present simulation results to validate the accuracy of our models and insights; our results show that the proposed models can accurately predict the E2E entanglement throughput under various settings. Moreover, our proposed approach to maximizing the throughput subject to required fidelity can deliver higher throughput of useful E2E entanglement compared to a baseline case with no MHTs.

The remainder of this paper is as follows. Section II provides a brief overview of related work. Our fidelity model and problem formulation are given in Section III. Section IV details our queueing models for the entanglement throughput estimation. Section V presents the optimal MHTs of entangled memories to maximize throughput with guaranteed fidelity. Section VI illustrates the efficacy of our approach via extensive simulations. Finally, Section VII offers concluding remarks.

\myitem{Notation:} $\mathbb{E}[X]$ denotes the expected value of a random variable $X$. $\mathrm{Exp(\lambda)}$ denotes an exponential distribution with parameter $\lambda>0$. For a vector $x$, $x^T$ denotes its transpose. 

\section{Related Work}\label{sec_related_work}

Many studies consider a quantum repeater chain with identical links and examine the probability distribution of waiting time \cite{shchukin2019waiting, simon2007quantum}. A similar model is used in \cite{brand2020efficient, li2021efficient}, but the authors develop polynomial-time algorithms for finding the distributions of waiting time and fidelity of E2E entanglement using the Werner state formulation. These distributions are computed up to a predefined truncation to avoid exponential complexity in \cite{shchukin2019waiting}, which are then used to design repeater cutoff times to maximize the secret key rate. 

The authors of \cite{vardoyan2019stochastic} present a continuous-time Markov chain (CTMC) model for a quantum switch in a star topology, where all links have the same number of qubit buffers, and define the capacity as the maximum possible switching rate when each link has a constant entanglement generation rate. This model also includes decoherence-associated cutoff time for qubit storage but does not consider entanglement noise models and cutoff time optimization.

Recent studies \cite{zubeldia2026matching, fittipaldi2023linear} build on the discrete-time model in \cite{vardoyan2019stochastic} for a quantum switch based on a matching queues network and consider general arrival distributions of entangled qubits. Here, the focus is on the stability of the switch under random request arrivals and analyzing the throughput optimality of the max-weight scheduling policy~\cite{Andrews2007}.

As swapping is an important operation for extending entanglement distance, there have been several swapping policies developed for a repeater chain.  They include, for example, sequential \cite{li2022connection}, doubling \cite{briegel1998quantum, van2013path, caleffi2017optimal}, swap-asap \cite{farahbakhsh2022opportunistic, kamin2023exact}, and heuristics \cite{ghaderibaneh2022efficient, chang2022order, haldar2024fast, mai2025towards}. Most of them are based on discrete-time models with homogeneous quantum memories, and the analysis focuses on the entanglement throughput and fidelity is handled via cutoff time of repeaters.

Our problem is closely related to double-ended queues with renewal arrivals and abandonment (see Section~\ref{subsec_CTMC}). Double-ended queues arise in many applications, e.g., taxi service systems, assembly lines, buyers and sellers in a commodity market, and organ transplant systems; \cite{perry1999perishable} considers an inventory system for perishable commodities with finite shelf size and waiting room for demands, where arrival rates for items and for demands are state-dependent and their waiting times can be either constant or exponentially distributed. The study in \cite{liu2015diffusion} examines fluid and diffusion approximations of queue length for different arrival or traffic patterns. 

We use double-ended queues to model the system state, including the number of entangled qubits, while taking into account nonnegligible swapping times and classical communication delays. Our model allows us to approximate the distribution of E2E entanglement fidelity and design optimal MHTs of quantum memories to maximize the E2E entanglement throughput for a given required fidelity. 



\section{Network Model and Problem Formulation}
    \label{sec:Network_Model}

\begin{figure}
    \centering
    \includegraphics[width=0.9\linewidth]{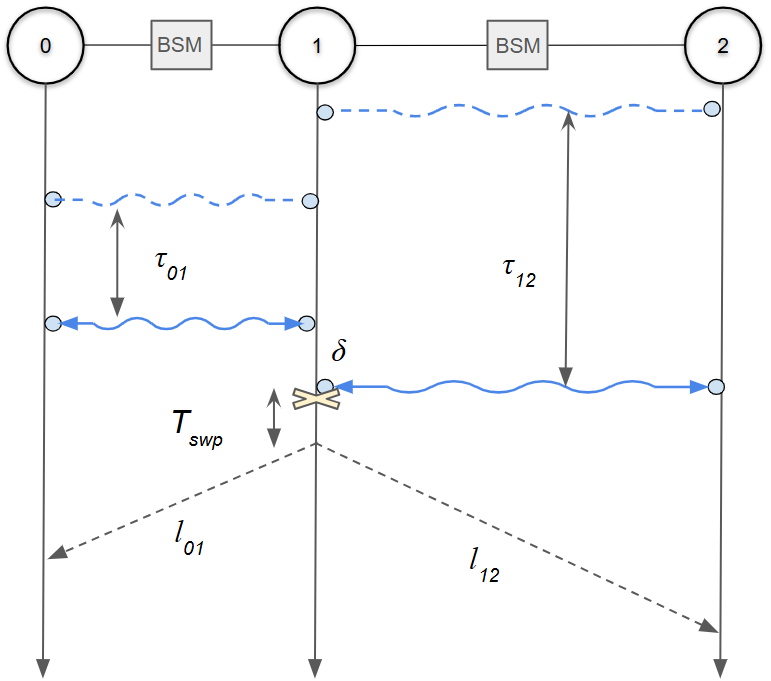}
    \caption{\small Timeline of a path with 3 quantum nodes. Each quantum channel is equipped with a photonic BSM station in the middle. Dashed and solid wavy lines indicate when entanglement generation attempts start with photon emissions and succeed with heralds acknowledged, respectively. Here, $\delta$ is the waiting time to swap, $T_{\rm swp}$ the swapping time, and $\ell_{ij}$ one-way latency between node $i$ and node $j$.}
    \label{fig_3n_path}
\end{figure}


Consider the network in Fig.~\ref{fig_3n_path} with a 3-node quantum path. Each quantum node has a limited number of memories for Bell pair creation. All nodes are connected classically, possibly over quantum channels (using time or frequency multiplexing). An edge $(i, j)$ means that nodes $i$ and $j$ share quantum channels for qubit transmission. For simplicity, we focus on entanglement generation and swapping, but do not consider entanglement purification here.

\vspace{0mm}
\subsection{Qubit Noise Models}
We focus on the two most common noise models: dephasing and depolarization; see, e.g., \cite[Chap.~8]{nielsen2010quantum}. These channels are typically assumed to act independently and sequentially. Since both are completely positive trace-preserving maps, their compositions are also valid quantum channels.
Let $\mathcal{E}_{v}^{(1)}$ and $\mathcal{F}_{w}^{(1)}$ denote the 1-qubit  depolarizing and dephasing channels with parameters $v$ and $w$, respectively, given by 
\begin{align*}
\mathcal{E}_{v}^{(1)}(\rho) = v \rho + (1{-}v)\frac{I_2}{2}, ~~
\mathcal{F}_{w}^{(1)}(\rho) = \frac{1 {+} w}{2}\rho + \frac{1 {-} w}{2} \sigma_z\rho \sigma_z,
\end{align*}
where $\rho$ is a qubit state, and $\sigma_z$ is the Pauli-Z matrix. The composition rules below apply (superscript omitted):
\begin{equation*}
\mathcal{E}\circ  \mathcal{F} \equiv \mathcal{F}\circ  \mathcal{E},~~ \mathcal{E}_{v_1}\circ  \mathcal{E}_{v_2} = \mathcal{E}_{v_1v_2}, ~~ \mathcal{F}_{w_1}\circ  \mathcal{F}_{w_2} = \mathcal{F}_{w_1w_2}.
\end{equation*}

Next, for a 2-qubit system under \textit{independent noises}, the tensor product of 1-qubit depolarizing (resp. dephasing) channels is  $\mathcal{E}_{v_1,v_2}^{(2)} = \mathcal{E}_{v_1}^{(1)}\otimes \mathcal{E}_{v_2}^{(1)}$ (resp. $\mathcal{F}_{w_1, w_2}^{(2)} = \mathcal{F}_{w_1}^{(1)}\otimes \mathcal{F}_{w_2}^{(1)}$) with parameters $(v_1,v_2)$ (resp. $(w_1,w_2)$). Note that
\begin{align}
&( \mathcal{E}_{v_1}^{(1)} \circ \mathcal{F}_{w_1}^{(1)} )
\otimes
( \mathcal{E}_{v_2}^{(1)} \circ \mathcal{F}_{w_2}^{(1)} ) \equiv  \mathcal{E}_{v_1, v_2}^{(2)} \circ \mathcal{F}_{w_1, w_2}^{(2)}.\label{eq_2q_composition}
\end{align}

We are interested in the action of 2-qubit channels on a Bell state, e.g., $
\ket{\Psi^+} = \frac{1}{\sqrt{2}}\bigl(\ket{01} + \ket{10}\bigr)$ with density matrix $ \rho_{\Psi^+} = \ket{\Psi^+}\!\bra{\Psi^+}$. 
Specifically, we have
\begin{align*}
\mathcal{F}_{w_1, w_2}^{(2)} (\rho_{\Psi^+})
&= \frac{1 + w_1w_2}{2}\rho_{\Psi^+} + \frac{1 - w_1w_2}{2}\rho _{\Psi^-}\\
\mathcal{E}_{v_1, v_2}^{(2)}(\rho_{\Psi^+}) 
&= v_1v_2 \rho_{\Psi^+} + (1-v_1v_2)\frac{I_4}{4}. 
\end{align*}
Dephasing mixes a Bell state with its phase-flipped state, while depolarization mixes it with the maximally mixed state $I_4/4$. Both reduce coherence and entanglement, but depolarization is uniform, leading to a fully mixed state at high noise. Crucially, the mixing effect for both depends only on the product of individual parameters. Thus, starting from the state $\ket{\Psi^+}$, independent dephasing or depolarization of the pair can be treated as single qubit noise with a combined parameter: 
\begin{equation*}
    \mathcal{F}_{w_1 w_2}^{(2)} (\rho_{\Psi^+})\equiv \mathcal{F}_{w_1, w_2}^{(2)} (\rho_{\Psi^+}),\quad \mathcal{E}_{v_1 v_2}^{(2)}(\rho_{\Psi^+}) \equiv  \mathcal{E}_{v_1, v_2}^{(2)}(\rho_{\Psi^+}).
\end{equation*}
Note further that these states are Bell-diagonal and remain so under any composition of $\mathcal{E}^{(2)}$ and $\mathcal{F}^{(2)}$.

\vspace{0mm}
\subsection{Quantum Memory} 
Quantum memory platforms vary in suitability for quantum communication. Superconducting qubits have short dephasing and relaxation times ($\mu\text{s}$ to $\text{ms}$), while spin-qubits and trapped-ions have longer times (ms to seconds or more), making them preferable for long ranges. We focus on the latter group for communication qubits, where dephasing is the dominant effect.

We use a combined dephasing and depolarizing channel to model quantum memory noise, with depolarization accounting for other errors (e.g., gate errors). Qubit decoherence at node $i$ over time $t$ is characterized by time-varying depolarizing and dephasing parameters $(v_i(t), w_i(t))$ as 
$$\mathcal{D}_i(t) = \mathcal{E}_{v_i(t)}^{(1)}\circ\mathcal{F}_{w_i(t)}^{(1)}.$$
Under Markovian conditions (e.g., fast fluctuations), parameters decay exponentially as $v_i(t)=e^{-\mu_it}, w_i(t) = e^{-\gamma_i t}$, where $(\mu_i, \gamma_i)$ are depolarizing and dephasing rates. Slow fluctuations, however, may lead to Gaussian decays  $w_i(t)=e^{-\gamma_i^2t^2}$ instead \cite{haffner2008quantum}. Here, we only assume these parameters decay over time. 
Finally, note that time-dependent depolarizing noise is important, particularly when extending memory dephasing time using dynamical decoupling techniques \cite{viola1999dynamical, biercuk2009optimized,  ezzell2023dynamical}. 
Dynamical decoupling introduces periodic qubit rotations, which cause additional bit-flip errors, modeled as time-dependent depolarizing noise. 


\vspace{0mm}
\subsection{Quantum Link}
We consider heralded entanglement generation (HEG) between \textit{emissive} quantum memories using \textit{photonic entanglement swapping} via a BSM station at midpoint of each quantum channel. Our abstract model focuses on entanglement generation and swapping, while acknowledging that practical details (e.g., memory hardware, frequency conversion, flying qubit encoding) are important; here we limit ourselves to two-photon interference with dual-rail encoding.

Entanglement generation attempts between two nodes over lossy channels like fiber have low success probabilities, which depend on hardware efficiency and exponential decay with channel length due to loss. A successful attempt, heralded by a desired BSM outcome, establishes an \textit{elementary entanglement pair} or \textit{quantum link} between them. Since success is rare, they can be modeled using Poisson processes. 

To describe the quality of an elementary entangled pair, suppose that the entangled state generated by a memory at an end node $s \in \{i,j\}$ is given by a perfect Bell state
\begin{equation}
    \ket{\Psi^+}_{s,s'} =  \big(\ket{1}_{s} \ket{0}_{s'} +  \ket{0}_s \ket{1}_{s'} \big)/{\sqrt{2}}, \quad s=i,j
\end{equation}
where $\ket{0}_{s'}$ and $\ket{1}_{s'}$ are the logical qubit states for the emitted optical modes transmitted from node $s$ to the BSM station; physical flying qubits depend on the choice of photonic encoding format. These optical modes must be indistinguishable to interfere at the BSM, which projects two memory qubits onto one of four entangled Bell states, thereby entangling them. 

Assuming standard linear optics, maximum BSM efficiency is~$0.5$ as only $\Psi^\pm$ can be detected \cite{calsamiglia2001maximum} (higher efficiencies are possible with more advanced setups using ancilla photons \cite{bayerbach2023bell} or nonlinear elements). Due to decoherence and imperfect BSM, we model heralded memory-memory state $\rho_{ij}(t)$ as
\begin{align}
    \rho_{ij}(t)&= \big( \mathcal{D}_{i}(t)\otimes \mathcal{D}_{j}(t)\big) \circ \big( \mathcal{E}_{v_{ij}^{\rm B}}^{(2)} \circ \mathcal{F}_{w_{ij}^{\rm B}}^{(2)} \big) (\rho_{\Psi^+}), \label{eq_state_after_photonicBSM}
\end{align}
which incorporates memory decoherence $\mathcal{D}_i(t) \otimes \mathcal{D}_j(t)$ and BSM errors, represented by depolarizing $v_{ij}^{\rm B}$ and dephasing $w_{ij}^{\rm B}$ parameters. 
Here, $w_{ij}^{\rm B}$ is the visibility of the Hong-Ou-Mandel (HOM) interference at the beam-splitter of the BSM, which depends on the indistinguishability degree between two photons, including their temporal, spatial, and spectral overlaps. 
Second, $v_{ij}^{\rm B}$ quantifies the excess noise from the quantum channels to the detector electronics that can cause accidental detections, including background photons (due to spontaneous Raman scattering from classical/control signals, noisy quantum frequency conversion) and detector dark counts during a detection window. It also covers uncontrolled fast-varying phase noise causing bit-flips in flying qubits. These parameters can be estimated via link characterization and depend on the fiber length as well. See~\cite{dhara2023entangling} for an alternative, detailed description of the state without memory decoherence. 


Note also that if the BSM outcome is not the desired state $\Psi^+$, corrective operations are required, e.g., a \mbox{Pauli-Z} gate for a $\Psi^-$ outcome in linear optical BSM. These operations, which can introduce noise (modeled as small 1-qubit depolarization), are not modeled explicitly for simplicity, as their effect can be absorbed into the depolarizing parameter of $\mathcal{D}_i$ or $\mathcal{D}_j$.


\myitem{Link Rate vs. Fidelity:} In general, these BSM errors can be suppressed using filters and reduced BSM detection windows. But, this reduces the elementary entanglement rate. A similar tradeoff exists with purification when multiple memory pairs are used. More studies are needed to compare the benefits of tuning the detection window versus purification on each link, which requires multiple memory pairs. This crucial tradeoff must be considered for each setting. For our modeling, we assume elementary entangled pairs for swapping are Bell states subjected to a combined dephasing and depolarizing channel.

\vspace{0mm}
\subsection{Memory Entanglement Swapping}
When entangled pairs $(0,1)$ and $(1,2)$ are available, node~$1$ performs entanglement swapping on its corresponding local memory qubits to establish an entangled pair $(0,2)$. This process takes~$T_{\rm swp}$ time and can be probabilistic. Swapping methods vary by platform: photonic BSM is fast but probabilistic (e.g., $0.5$ efficiency), while trapped-ion swapping (using 2-qubit Molmer-Sorensen gate and readout) is often deterministic but slower (e.g., $\approx 2\,\textrm{ms}$ for the electron shelving setup in \cite{krutyanskiy2023telecom}) to achieve high fidelity. 
As errors can also occur during swapping, we model them using a 2-qubit depolarizing channel with a fixed parameter $\beta_1$ to capture fidelity loss. While dephasing is negligible during swapping itself, it still needs to be accounted for during qubits' waiting time prior to swapping. We assume independent dephasing for memories within a node and leave potential collective dephasing for future work. 




After each swap, end nodes are heralded using classical communication (taking $\ell_{01}$ and $\ell_{12}$). Like photonic BSMs, success may require local corrections to achieve the desired Bell pair. If an entanglement is acknowledged, end nodes may wait $\tau_{\rm app}$ time to consume it before releasing the memories; otherwise (upon swap failure), memories are released immediately after the heralding signal. All such waiting periods cause independent decoherence of remote qubits. 

\vspace{0mm}
\subsection{E2E State and Fidelity}
To study the maximum rate at which E2E entanglements can be delivered, we assume a negligible application time $\tau_{\rm app}$. The state of E2E entangled pair at time $t$ post-swapping is
\begin{align*}
&\big( \mathcal{E}_{v_0(t)v_2(t) v_{02}^{\rm B}}^{(2)} \circ \mathcal{F}_{w_0(t)w_2(t) w_{02}^{\rm B}}^{(2)} \big) (\rho_{\Psi^+}) \ \text{with}\\
    &v_{02}^{\rm B} = v_{01}^{\rm B} v_{12}^{\rm B} v_1(a_{1L})v_1(a_{1R}) \beta_1 \\
    &w_{02}^{\rm B} = w_{01}^{\rm B} w_{12}^{\rm B} w_1(a_{1L})w_1(a_{1R}),  
\end{align*}
where  $a_{1L}$ and $a_{1R}$ are the ages of qubits at node $1$ when swapped. 
Limiting the entangled qubits at both ends of the path to the arrivals of heralding signals yields
\begin{align}
    \!\!\!\rho_{02} &=  v_{02}\big( \frac{1 + w_{02}}{2}  \rho_{\Psi^+} + \frac{1 - w_{02}}{2} \rho_{\Psi^-}\!\big) + (1-v_{02})\frac{I_4}{4} \nonumber\\
    v_{02} &= v_0(a_{0R})v_2(a_{2L}) v_{02}^{\rm B}, \quad
    w_{02} = w_0(a_{0R})w_2(a_{2L}) w_{02}^{\rm B}, \nonumber
\end{align} 
where $a_{0R}$ and $a_{2L}$ are the ages of memory qubits at nodes 0 and
2, respectively. Thus, the E2E fidelity is 
\begin{equation}
    f_{e2e} = \bra{\Psi^+}\rho_{02}\ket{\Psi^+} = \frac{1}{4}\big(1+v_{02}(1+2w_{02}) \big). \label{e2e_fidelity}
\end{equation}
Here, the effect of depolarization is included in $v_{02}$, and that of dephasing is captured in $w_{02}$; both show memory decoherence due to accumulating memory qubit ages.

\vspace{0mm}
\subsection{Problem Formulation}
Under the above noise model, the E2E entanglement fidelity decreases with the waiting times for swapping entangled pairs; longer waiting times lead to lower fidelity. Thus, to guarantee a minimum fidelity level, it is necessary to impose cutoff times on the quantum memories for entangled pairs, which may differ between the two links due to their distinct properties. Our goal is to optimize these quantum memory cutoff times, i.e., MHTs, to maximize the E2E entanglement throughput while adhering to a specified minimum fidelity constraint. This approach allows us to characterize \textit{the swapping capacity of a quantum repeater as the trade-off between the achievable rate and the resulting fidelity}.

\section{E2E Entanglement Throughput Models}
    \label{sec_throughput}
In this section, we describe analytical models for approximating the throughput for a path with two links shown in Fig.~\ref{fig_3n_path}. We refer to edge ($i$, $i+1$) simply as link $i$, $i = 0, 1$. We start with assumptions and then consider the simpler cases with a unit memory for each link, followed by more general cases with multiple memories. 

\vspace{0mm}
\subsection{Assumptions and Setup}
First, we introduce the following assumptions on the generation and storage of elementary entanglements.

{\bf A1.} When there are free memories available at the end nodes of a link, elementary entangled pairs or simply entanglements are generated on the link according to an independent Poisson process with a possibly state-dependent rate (i.e., available quantum resources at both ends of the links). 

{\bf A2.} Elementary entanglements can be stored in memories for a certain amount of time before they decohere or become unusable. We will handle a given fidelity requirement on provided E2E entanglements by imposing a limit $W_i$ on the holding times of elementary entanglements on link $i$.

{\bf A3.} When there is a match between two elementary entanglements on the two links, swapping takes place immediately. It takes $T_{\rm swp}$ and is successful with probability $q \in (0, 1]$ independently of previous attempts. Its outcome is heralded to the end nodes 0 and 2, which takes $\ell_{01}$ and $\ell_{12}$, respectively.
For simplicity, we assume that the end nodes consume the entangled qubits as soon as they are heralded and then release the memories for new entanglement generation attempts.  


{\bf A4.} Generally, a fixed node initiates every link attempt. For example, if node~1 starts link~0, a new negotiation can piggyback on heralding messages, allowing an immediate start. On the other hand, if node 0 initiates, a delay of $\ell_{01}$ is incurred due to heralding latency. To account for this, we assume that after entanglement swapping, link $i$ can only attempt a new generation with swapped qubits after a delay of $T_i \ge 0$. This delay also incorporates hardware utility cycles, which are assumed to be small compared to communication latency.

\vspace{3pt}
Let us describe each link $i$ using a tuple 
\begin{equation}
    \mathcal{Q}_i = (K_i, \boldsymbol{\lambda}_i, W_i, T_i), \quad i = 0, 1, 
\end{equation}
where $K_i$ denotes the memory capacity, and $\boldsymbol{\lambda}_i \in \mathbb{R}^{K_i}$ is the vector of elementary entanglement arrival rates at different states described below. 
To study the E2E entanglement throughput, we adopt a queueing model to capture the number of elementary entanglements stored in memory for each link;  we refer to the queue associated with link $i$ as $\mathcal{Q}_i$ and say that an entanglement $e_i$ arrives when a new elementary entangled pair is generated on link $i$. When a new entangled qubit is stored in a memory (at both ends of a link), the queue length increases by one. On the other hand, when an entangled qubit is swapped at the repeater or expires, thereby freeing up a memory, the queue length decreases by one. Hence, the length of queue $i$ captures the number of link $i$ elementary entanglements stored in the memories. 


We will consider two cases: (i) unit-memory queues and (ii) multiple-memory queues. The main reason for considering case (i) separately is that it is a canonical setup for realizing long-distance entanglement distribution via swapping, where the dynamics can be modeled accurately. 

\vspace{0mm}
\subsection{Unit-Memory Queues}\label{subsec_3node_path_unit_mem}
In this case, the repeater is equipped with two memories and allocates one memory to each link with $K_0= K_1= 1$. Since there is only one memory available on each link, when $\mathcal{Q}_i$ is empty, elementary entanglements $e_i$ arrive according to a Poisson process with rate $\lambda_i > 0$ ({\bf A1}). An entanglement $e_i$ stored in memory remains valid for a maximum duration of $W_i$. If no swapping occurs with an entanglement $e_{1-i}$ within $W_i$, $e_i$ expires and is discarded ({\bf A2}). After swapping or expiration of an elementary entanglement, it is removed and its queue is reset. After undergoing a reset delay of $T_i$, entanglement generation then resumes on $\mathcal{Q}_i$ ({\bf A4}).

We analyze E2E entanglement throughput using renewal reward theory~\cite{Ross}, where a renewal occurs when \textit{both} queues become empty, starting a cycle. A reward of 1 is assigned if an E2E entanglement is generated in a cycle, and 0 otherwise. The exact analytical throughput expression is rather complex and deferred; here we provide an approximation, which is accurate when reset delays are small compared to MHTs ($W_i, i = 0, 1$) and inter-arrival times. It assumes that when one queue resets upon expiration, the other also resets (pausing arrivals by a small delay), creating shorter ``resetting cycles."



\subsubsection{Approximate Cycles and Swapping}
Define $\Lambda := \lambda_0  + \lambda_1$ as the aggregate arrival rate of the elementary entanglements. Suppose that the reset delays $T_i, i = 0, 1,$ are much smaller than $\Lambda^{-1}$ and that the probability that an entanglement arrives during a reset period is small. In this case, we can assume that a renewal occurs after swapping of entanglements or expiration of a stored entanglement in a queue. 

Suppose that an entanglement $e_i$ arrives at queue $i$ and must wait for an entanglement $e_{1-i}$ for swapping. If its MHT is reached before an entanglement $e_{1-i}$ arrives, the cycle ends without a match. Otherwise, a match happens upon the arrival of an entanglement $e_{1-i}$ and swapping takes a fixed amount of time $T_{\rm swp}$ and is successful with probability $q$. In this case, the amount of time it takes to complete the swapping and be ready for new entanglement generations is approximately
\begin{align}
    T_p := T_{\rm swp}+\max(T_0, T_1)  \ .
\end{align}
Let $\alpha_i = 1-e^{-\lambda_{1-i}W_i}$ be the probability that, conditional on the event that an entanglement $e_i$ arrives first during a cycle, an entanglement $e_{1-i}$ arrives within $W_i$ to trigger swapping. 
\begin{theorem}
The following expression provides an approximate E2E entanglement rate or throughput: 
\begin{equation}
R_{um} = \frac{q(\alpha_0 \lambda_0 + \alpha_1 \lambda_1)}{1+\sum_{i=0}^{1}\alpha_i\lambda_i(T_p+\lambda_{1-i}^{-1}) + (1-\alpha_i)\lambda_iT_i} 
    \label{eq_Rate_UnitCapacity}
\end{equation}
\end{theorem}
\begin{proof}
Let $T_{\text{cycle}}$ be the duration of a cycle (which is a random variable) and $p_{\text{match}}$ the probability of having a match in a cycle. 
Thus, the expected reward of a cycle is $q \cdot p_{\text{match}}$. Then, the renewal reward theorem~\cite{Ross} tells us that, with probability 1 
the long-term reward rate is given by
\begin{align}
R_{um} = {q \cdot p_{\text{match}}}/{\mathbb{E}[T_{\text{cycle}}]} \ . 
    \label{eq:RRT}
\end{align}
We proceed to evaluate those terms. Let $j=1-i$ below.

\vspace{3pt}
$\bullet$ {$p_{\text{match}}$}: Let $G_i, i = 0, 1$, be the event that the first entanglement in a cycle is $e_i$.  
Since we assume that the entanglements arrive according to mutually independent Poisson processes, the probability of event $G_i$ is  $\frac{\lambda_i}{\Lambda}$. Conditional on $G_i$, the probability that an entanglement $e_{j}$ arrives within the MHT $W_i$ is $\alpha_i = 1-e^{-\lambda_{j} W_i}$ due to the memoryless property of an exponential distribution. Thus, by conditioning on the first link that generates an entanglement and using the law of total probability, we can compute the probability that swapping happens in a cycle, which is equal to  
\begin{equation}
p_{\text{match}} = (\lambda_0 \alpha_0+\lambda_1 \alpha_1)/{\Lambda}.
\end{equation}

\vspace{3pt}
$\bullet$ {$\mathbb{E}[T_{\text{cycle}}]$:}
A renewal cycle comprises of three periods.

i) \textit{Idle period:} Starting with empty queues, the time until the first entanglement arrival, which is the minimum between two independent exponential random variables with parameters $\lambda_0$ and $\lambda_1$, is exponentially distributed with parameter $\Lambda$ and its expected value is $\Lambda^{-1}$.
    
ii) \textit{Waiting period:} If the first arrived entanglement is $e_i$, its qubits will be stored in memories until either an entanglement $e_{j}$ arrives in time or it expires after its MHT $W_i$. Thus, its waiting time can be modeled using the minimum of $Y_{j} \sim \mathrm{Exp}(\lambda_{j})$ and the fixed MHT $W_i$, and the expected waiting time $\mathbb{E}[\min(Y_{j}, W_i)]$ is given by 
\begin{align}
\int_0^{W_i} \!\! y \lambda_j e^{-\lambda_j y} dy 
    + \int_{W_i}^\infty \!\! W_i \lambda_j e^{-\lambda_j y} dy 
    = \frac{1 {-} e^{-\lambda_j W_i}}{\lambda_j} = \frac{\alpha_i}{\lambda_j}.\nonumber
\end{align}
    
iii) \textit{Swapping/reset period:} Suppose that an entanglement $e_i$ arrives first. If an entanglement $e_{j}$ arrives before $W_i$ (with probability $\alpha_i$), swapping is performed and both queues are reset in approximately $T_p=T_{\rm swp}+\max\{T_0, T_1\}$. On the other hand, if $e_i$ expires instead (with probability $1-\alpha_i$), both queues are reset after $T_i$. Thus, conditional on $G_i$, the expected value of the swapping/reset period is $\alpha_i T_p + (1-\alpha_i) T_i$. 

Now let $H_i = \alpha_i(T_p+\lambda_{j}^{-1}) + (1-\alpha_i)T_i$, $i = 0, 
1$. From the above discussion, the expected duration of a cycle can be computed by conditioning on the first arriving entanglement:
\[
\mathbb{E}[T_{\text{cycle}}] = (1 + \lambda_0 H_0 + \lambda_1 H_1)/{\Lambda}.
\]
Substituting $p_{\text{match}}$ and $\mathbb{E}[T_{\text{cycle}}]$ in \eqref{eq:RRT} yields \eqref{eq_Rate_UnitCapacity} as desired.
\end{proof} 

When the entanglement MHTs are sufficiently large and $1-\alpha_i = e^{-\lambda_{1-i} W_i} \approx 0$, stored entanglements will rarely expire and thus be swapped with a probability close to 1. In this case, the throughput is approximately given by  
\[ 
R_{um} \approx {q \Lambda}/
    {(1+\lambda_1\lambda_0^{-1}  +\lambda_0\lambda_1^{-1}+T_p\Lambda)}.
\]
In addition, if classical delays are small and entanglement generation rates are low so that $T_p \ll \min (\lambda_0^{-1}, \lambda_1^{-1} )$, then $R_{um} \approx q\,\Lambda/(1+\lambda_0 \lambda_1^{-1} + \lambda_1 \lambda_0^{-1})$. Finally, if two links are identical and $\lambda_0 = \lambda_1 = \lambda$, then $R_{um} \approx \frac{2}{3}\lambda q$, which reduces to the well-known fact that if $X_1, X_2 \sim \mathrm{Exp}(\lambda)$ are independent, $\mathbb{E}[\max(X_1, X_2)] =1.5/\lambda$. 
    
In the other extreme case when MHTs $W_i$, $i = 0, 1$, are small compared to $\min (\lambda_0^{-1}, \lambda_1^{-1} )$, the assumption that the entanglement generations can be approximated using Poisson processes becomes less accurate and our model tends to overestimate the E2E throughput. While it is possible to correct this using a different arrival distribution, the analysis becomes more involved and we omit its discussion here. 



\vspace{0mm}
\subsection{Multiple-Memory Queues}\label{subsec_CTMC}
When multiple memories are available for storing entangled pairs on each link, the system need not reset after each swapping because there can be other entangled pairs stored in the memories for one of the two links. As a result, we can no longer use renewal reward theory.  
Instead, we consider the system as a double-ended queue with possibly state-dependent arrival rates, departure rates, and post-processing delays. Specifically, we approximate its dynamics using a CTMC as described below.

\subsubsection{CTMC Model} 
We consider queues with state-dependent Poisson arrival processes. If $\mathcal{Q}_i$ has $n < K_i$ occupied pairs, its arrival rate is $\lambda_{i|n}>0$. If it is full, $\lambda_{i|K_i} = 0$ and no more entanglements can be generated until a memory is freed up. 

To facilitate analysis, instead of dealing with fixed entanglement MHTs, we approximate them using independent exponential random variables; entangled pairs stored in $\mathcal{Q}_i$ have exponentially distributed MHTs with parameter $\theta_i = 1/W_i$. Similarly, we also model the swapping time and reset delays as independent exponential random variables with parameters $\mu_p = 1 / T_{\rm swp}$ and $\mu_i = 1/T_i$, $i = 0, 1,$ respectively. 
%
Under these Markovian approximations, our model only needs to keep track of the number of entanglements that are either available for swapping or currently undergoing swapping and the number of memories that are being reset following swapping or an expiration of a stored entangled pair. 

\vspace{3pt}
\noindent\textbf{State space}:
Let $(n,m_p,m_0, m_1)$ be the system state, where
\begin{itemize}
    \item $n\in \mathbb{Z}$: the net number of waiting entangled pairs; 
    $n<0$ (resp. $n>0$) means that there are $|n|$ stored entangled pairs in $\mathcal{Q}_0$ (resp. $\mathcal{Q}_1$) and none in $\mathcal{Q}_1$ (resp. $\mathcal{Q}_0$).
    
    \item $m_p \in \mathbb{Z}_{+}$: the number of pairs that are being swapped
    
    \item $m_i \in \mathbb{Z}_{+}$: the number of pairs
    being reset on $\mathcal{Q}_i, i= 0, 1.$
\end{itemize}
Let $N_i$ denote the queue lengths, i.e., 
\begin{equation*}
N_0 = \max(0,-n) + m_p + m_0, \quad
N_1 = \max(0,n) + m_p + m_1.
\end{equation*}
The finite state space of the CTMC is given by
$$
\mathcal{S} = \left\{ (n, m_p, m_0, m_1) \in \mathbb{Z}\times\mathbb{Z}_+^3 \; | \;
N_i \le K_i,\ i=0,1
\right\}.
$$

\begingroup
\setlength{\tabcolsep}{4pt} 
\renewcommand{\arraystretch}{1.15} 
\begin{table}
    \centering
    \begin{tabular}{l|c|l|c}
    \toprule
        event & cond. & next state (if feasible) & rate \\
    \midrule
        queue 0 arrival, no swap & $n \le 0$ & $(n{-}1, m_p, m_0, m_1)$       & $\lambda_{0|N_0}$  \\
        queue 0 arrival, swap    & $n > 0$   & $(n{-}1, m_p{+}1, m_0, m_1)$   & $\lambda_{0|N_0}$ \\
        queue 1 arrival, no swap & $n \ge 0$ & $(n{+}1, m_p, m_0, m_1)$       & $\lambda_{1|N_1}$ \\
        queue 1 arrival, swap    & $n < 0$   & $(n{+}1, m_p{+}1, m_0, m_1)$   & $\lambda_{1|N_1}$  \\
        an $e_0$ expires      & $n < 0$   & $(n{+}1, m_p, m_0{+}1, m_1)$     & $|n|\theta_1$ \\
        an $e_1$ expires      & $n > 0$   & $(n{-}1, m_p, m_0, m_1{+}1)$     & $n\theta_2$  \\
        a queue 0 memory resets            & $m_0 > 0$   &  $(n, m_p, m_0{-}1, m_1)$      & $m_1\mu_1$ \\
        a queue 1 memory resets            & $m_1 > 0$   &  $(n, m_p, m_0, m_1{-}1)$      & $m_2\mu_2$ \\
        pairing finished        & $m_p > 0$   &  $(n, m_p{-}1, m_0{+}1, m_1{+}1)$      & $m_p\mu_p$ \\
    \bottomrule
    \end{tabular}
    \caption{\small Transitions from state $(n,m_p,m_0, m_1)$}
    \label{table_CTMC_state_transition}
\end{table}
\begingroup


\noindent\textbf{State transitions}: All feasible transitions from each state $(n,m_p,m_0, m_1) \in \mathcal{S}$ are listed in Table~\ref{table_CTMC_state_transition}. Since every state communicates with the state ${\bf 0}$, this CTMC is irreducible. Let $Q$ be the generator matrix of the CTMC that contains the transition rates. 
The following is a standard result~\cite{GrimmettStirzaker}.
\begin{lemma}
The steady-state distribution $\boldsymbol{\pi} = (\pi_s: s\in \mathcal{S})$ of an irreducible CTMC can be found by solving 
\begin{equation} 
\boldsymbol{\pi}Q = {\bf 0}, \quad \textstyle\sum_{s \in \mathcal{S}} \pi_s = 1. \label{CTMC_steady_state_prob}
\end{equation}
\end{lemma}

Consider the states that, following a feasible transition, can trigger new swapping of two entanglements, i.e., states with waiting entanglements in queue $i$ while queue $1-i$ is not full. These are the states in $\mathcal{S}_0 \cup \mathcal{S}_1$, where 
\begin{align*} 
\mathcal{S}_0 &= \left\{ (n, m_p, m_0, m_1) \in \mathcal{S} \ | \
n<0,N_1 < K_1 \right\} \\
\mathcal{S}_1 &= \left\{ (n, m_p, m_0, m_1) \in \mathcal{S} \ | \
n>0, N_0 < K_0 \right\}. 
\end{align*}
The rate at which swapping occurs at the states in $\mathcal{S}_i$ is the sum of probabilities of the states multiplied by the entanglement generation rate of link $1-i$. As a result, we have the following.
\begin{theorem} 
The throughput of the CTMC is given by 
    \begin{equation}
        R_{mm} = \textstyle q \big( \sum_{s\in \mathcal{S}_0} \pi_{s} \cdot\lambda_{1|N_1}  +  \sum_{s\in \mathcal{S}_1} \pi_{s} \cdot\lambda_{0|N_0} \big).
    \end{equation}
\end{theorem}
This tells us that computing the E2E entanglement throughput reduces to solving \eqref{CTMC_steady_state_prob} for the steady-state distribution $\boldsymbol{\pi}$ of the CTMC. This is computationally feasible even for a reasonably large system. 

\subsubsection{Simplified Settings}
\label{subsubsec_Rate_small_T}
When reset delays $T_0, T_1$ and swapping time $T_{\rm swp}$ are small compared to mean arrival times of the queues $1/\lambda_{i,N_i}$, memories reset almost instantaneously after matching or expiration. Thus, $m_0, m_1$, and $m_p$ are equal to zero with high probability. As a result, we can approximate the state space using a one-dimensional integer line: $n \in \{-K_0, \dots, K_1\}$. This is a Birth-Death process (BDP) that is much simpler than the general case above; see, e.g., \cite{perry1999perishable, vardoyan2019stochastic} for a similar analysis. As a result, we omit the proof of the following corollary.

\begin{corollary}\label{cor_Rate_small_T}
When $T_p \lambda_{i|N_i}\ll 1$, $i = 0, 1$, the E2E entanglement throughput can be approximated using
    \begin{equation}\label{eq_Rate_small_T}
        R_{mm} \approx {q ( \lambda_{1|0}E_{0}+ \lambda_{0|0} E_{1}  )}/{(1+  E_{0} + E_{1})}
    \end{equation}
\[
\text{with}~E_{0}=\! \sum_{n=1}^{K_0} \prod_{k=1}^n \frac{\lambda_{0|k-1}}{\lambda_{1|0}+k\theta_0} ~~\text{and}~~ E_{1}= \!\sum_{m=1}^{K_1} \prod_{k=1}^m \frac{\lambda_{1|k-1}}{\lambda_{0|0}+k\theta_1}.
\]
\end{corollary}

With a change of variables $\tilde{\alpha}_i = \frac{\lambda_{i|0}}{\lambda_{1-i|0}}E_{1-i},$ 
\eqref{eq_Rate_small_T} becomes
$R_{mm} = q ( \tilde{\alpha}_0\lambda_{0|0}+ \tilde{\alpha}_1\lambda_{1|0})/(1+  \sum_{i=0}^1 \tilde{\alpha}_i\lambda_{i|0}/\lambda_{1-i|0})$. 
This is similar to the rate in \eqref{eq_Rate_UnitCapacity} for unit-memory queues except for the terms in the denominator of \eqref{eq_Rate_UnitCapacity}, which account for the reset delays. This is because, with unit-memory queues, no new entanglement can be generated during resets. However, with multiple-memory queues, new attempts can start immediately after swapping when another pair of memories is available. 





\section{Maximizing Throughput for a Given Fidelity}
\label{sec_maximize_Rate_Fidelity}
For ease of exposition, we first focus on Markovian dephasing memories with $w_i(t) = e^{-\gamma_i t}$ and then discuss the general case later. In this case, $w_{02} =w_{01}^{\rm B}w_{12}^{\rm B} e^{-A}$ with 
\begin{equation}
    A = \gamma_0a_{0R} + \gamma_2a_{2L} + \gamma_1(a_{1L}+a_{1R})
\end{equation}
which can be viewed as the total age of all qubits weighted by their dephasing rates. Recall that $a_{1L}$ and $a_{1R}$ are the ages of qubits at node $1$ when they are swapped. 

\vspace{0mm}
\subsection{Total Age of Qubits and Waiting Time} 
To understand better the stochastic nature of the delivered E2E entanglement fidelity, let us elaborate on the total age of all qubits: upon receiving heralding messages, we have $a_{0R} = a_{1L}+T_{\rm swp}+\ell_{01}$ and $a_{2L} = a_{1R}+T_{\rm swp}+\ell_{12}$. 
Thus, 
\begin{align}\label{eq_total_age_3node_path}
    {A} &=  \gamma_0(\ell_{01}+T_{\rm swp}) + \gamma_2(\ell_{12}+T_{\rm swp}) \nonumber\\
    &\quad + (\gamma_0+\gamma_1)a_{1L}+ (\gamma_1+\gamma_2)a_{1R}. 
    \end{align}
Note that the ages of the swapped qubits $a_{1L}$ and $a_{1R}$ depend on which entangled pair is heralded first at node $1$, and how long it has to wait for the other to be heralded. 
In our queueing model, this means that they depend on which queue experiences an arrival first and how long it  must wait for an arrival at the other queue. 
We denote this waiting time at the first queue by $\delta_1$. 
If the first arrival is at queue $0$, then $a_{1L} = \tau_{01} + \delta_1$ and $a_{1R} = \tau_{12}$. Similarly, if the first arrival is at queue $1$, $a_{1L} = \tau_{01}$ and $a_{1R} = \tau_{12}+\delta_1$, where $\tau_{ij}$ is the duration of a successful entanglement generation attempt on link $i$. In view of this, let $\gamma_{ij} = \gamma_i +\gamma_j$ denote the combined dephasing rate of the pair $(i,j)$ and rewrite ${A}$ as 
\begin{align}\label{eq_total_age_3node_path_RV}
{A} &= {A}_0 + Z \\
{A}_0 &= \gamma_0\ell_{01} + \gamma_2\ell_{12} + \gamma_{02}T_{\rm swp} + \gamma_{01}\tau_{01}+ \gamma_{12}\tau_{12}, \nonumber
\end{align}
where ${A}_0$ is fixed and $Z$ is a nonnegative random variable representing the random part of $A$: $Z{=}\gamma_{01}\delta_1$ (resp. $Z{=}\gamma_{12}\delta_1$) if the first arrival is at $\mathcal{Q}_0$ (resp. $\mathcal{Q}_1$). Because the waiting time $\delta_1$ is upper-bounded by the respective entanglement MHT in both cases, we can bound the total age as follows:
\begin{align}
    {A} \le {A}_0 + \max (\gamma_{01}W_{0}, \gamma_{12}W_{1}) =:{A}_{\max}\label{eq_Abar_max}.
\end{align}
This bound holds regardless of the distribution of $Z$ and can be used to enforce a minimum fidelity requirement as shown in the following subsections. 

Let us now analyze the total age ${A}$ more carefully using our models described in the previous section. Due to limited space, we consider here only the case with unit-memory queues; the multi-memory case is more involved as it depends also on the matching policy for swapping. The probability density function (PDF) of waiting time $\delta_1$ can be studied using a 2-queue model with
$
\mathcal{Q}_{0} {=} (1, \lambda_{0}, W_{0}, \ell_{01})~\text{and}~ 
\mathcal{Q}_{1} {=} (1, \lambda_{1}, W_{1}, \ell_{12}).
$
To simplify our analysis, we focus on the approximated system in Section~\ref{subsec_3node_path_unit_mem}, where the waiting time is a mixture of two independent {\em truncated exponential} random variables. 

Recall $\alpha_{i} = 1-e^{-\lambda_{1-i}W_{i}}, i = 0, 1$.
Let $S$ be the event that there is swapping in a cycle and $G_i, i = 0, 1$, the event that the first arrival is to queue $i$. 
The probability of the event $S$ is 
$
p_{\text{match}} =  (\lambda_{0}\alpha_{0}+\lambda_{1}\alpha_{1})/{\Lambda}, 
$
and the conditional probability $ p_{i | S}:=\Prob{G_i | S} = {\lambda_{i}\alpha_{i}}/{\Lambda  p_{\text{match}}}$. Thus,  
$$
p_{i|S}   
= {\lambda_{i}\alpha_{i}}/{(\lambda_{0}\alpha_{0}  + \lambda_{1}\alpha_{1})}, \quad i = 0, 1. 
$$
Conditional on the event $S \cap G_0$, we have $a_{1L} = \tau_{01}+\delta_1$ and $a_{1R} = \tau_{12}$, where the waiting time $\delta_1$ follows an exponential distribution with parameter $\lambda_{1}$ up to $W_0$. We denote this truncated distribution by $\mathrm{TruncExp}(\lambda_{1}, W_{0})$, whose PDF is 
\begin{align*}
g_{\mathtt{TE}}(x;\lambda_1, W_0) = 
 \frac{\lambda_1 e^{-\lambda_1 x}}{1-e^{-\lambda_1 W_0} } \mathbf{1}_{(0, W_0)}(x)
\end{align*}
In this case, $Z= \gamma_{01}\delta_{1L}$ with $\delta_{1L}\sim \mathrm{TruncExp}(\lambda_{1}, W_{0}).$
Similarly, in the event $S \cap G_1$, we have $Z= \gamma_{12}\delta_{1R}$ with $\delta_{1R}\sim \mathrm{TruncExp}(\lambda_{0}, W_1)$. 
Putting together, $Z$ follows a mixture of two truncated exponential distributions with PDF 
\[
g_Z(z) 
= \frac{p_{0|S}}{\gamma_{01}} g_{\mathtt{TE}}(z/\gamma_{01};\lambda_{1}, W_{0}) 
+ \frac{p_{1|S}}{\gamma_{12}} g_{\mathtt{TE}}(z/\gamma_{12};\lambda_{0}, W_{1}).
\]
This distribution of $Z$ allows us to compute that of the total age $A$ using the relationship in \eqref{eq_total_age_3node_path_RV}. For example, the expected age is given as follows with $h(\alpha) = \alpha + (1-\alpha)\ln (1-\alpha)$: 
\begin{align}
\mathbb{E}[{A}] &= {A}_0 + p_{0|S}\gamma_{01}\mathbb{E}[\delta_{1L}] + p_{1|S}\gamma_{12}\mathbb{E}[\delta_{1R}]\nonumber\\
&= {A}_0 + \frac{\gamma_{01}\lambda_{0}^2 h(\alpha_{0}) + \gamma_{12}\lambda_{1}^2 h(\alpha_{1})}{\lambda_{1}\lambda_{0}^2\alpha_{0} + \lambda_{0}\lambda_{1}^2\alpha_{1} }
    \label{eq_Abar_mean}
\end{align}


\vspace{0mm}
\subsection{E2E Fidelity Requirement and Maximum Holding Times}
As we can see, fidelity degradation is mainly caused by aging of all involved qubits, and the total age consists of fixed heralding latencies and random waiting time at the repeater. 
In the ideal case when there is no wait, i.e., $\delta_1=0$ and $A = {A}_0$, 
\begin{equation}\label{eq_Fmax}
F_{\max} = \big(1+v_{02}(1+2w_{01}^{\rm B}w_{12}^{\rm B} e^{-{A}_{0}}) \big)/4 
\end{equation}
is the maximum E2E fidelity without entanglement purification. It is clear that, in order to guarantee high E2E fidelity, we should use small MHTs to limit the total age. However, reducing MHTs, hence waiting time at the repeater, also decreases the probability of swapping $p_{\text{match}}$, thus reducing throughput. Therefore, the selection of MHTs $W_i, i = 0, 1$, presents a trade-off between fidelity and throughput.  

Consider a fixed fidelity requirement $F_{\rm req} < F_{\max}$. We are interested in delivering E2E entangled pairs with $f_{e2e} \geq F_{\rm req}$. Based on our analysis above, this constraint can be translated to an equivalent constraint on the total dephasing age $A$:
\begin{align}
    {A} \le \ln \frac{ 2 {v}_{02} w^{\rm B}_{01}w^{\rm B}_{12}}{4F_{\rm req}-1-{v}_{0n}} =: {A}_{\rm thr}
    \label{Age_threshold_e2e}
\end{align}
We can ensure that this constraint is satisfied by choosing appropriate MHTs $W_{i}, i = 0, 1$. Based on this observation, we formulate finding the optimal MHTs that maximize the E2E entanglement throughput as an optimization problem below.


\vspace{0mm}
\subsection{Optimal Entanglement Maximum Holding Times}
\label{sec_maximize_Rate_Fidelity}
First, we rewrite the rate at which swapping takes place in our queueing models in a unified form as follows:
\begin{equation*}
    R_m({\bf x}) = \frac{c_0 x_0 + c_1x_1}{1+x_0+x_1} \ , 
\end{equation*}
where ${\bf x} = [x_0, x_1]$ depends on the MHTs, and ${\bf c} = [c_0, c_1]$ depends only on fixed parameters. For unit-memory queues, 
    \begin{align*}
        {\bf x} &\!=\!\left[ 
        \begin{aligned}
        \frac{\alpha_{0}\lambda_{0}(T_p - \ell_{01} + \lambda_{1}^{-1})}{1+\lambda_{0}\ell_{01}+\lambda_{1}\ell_{12}},\  
        \frac{\alpha_{1}\lambda_{1}(T_p - \ell_{12} + \lambda_{0}^{-1})}{1+\lambda_{0}\ell_{01}+\lambda_{1}\ell_{12}}
        \end{aligned} 
        \right]\\
        {\bf c}&\!=\! \left[ {1}/{(T_p-\ell_{01} +\lambda_{1}^{-1})},\ {1}/{(T_p-\ell_{12} +\lambda_{0}^{-1})} \right], 
    \end{align*}
    follow from \eqref{eq_Rate_UnitCapacity}, 
    where $T_{p} = T_{\rm swp}+  \max(\ell_{01}, \ell_{12})$, and $\alpha_{i} = 1-e^{-\lambda_{1-i}W_{i}}$ depend on the MHTs. 
    For multiple-memory queues, let us use the closed-form in \eqref{eq_Rate_small_T} instead of the general case (which we leave for future work):
    \begin{align*}
        {\bf x} &\!\!=\! \left[ \sum_{n=1}^{K_{0}} \prod_{k=1}^n \frac{\lambda_{0|k-1}}{\lambda_{1|0}+k W_{0}^{-1}}, \sum_{m=1}^{K_{1}} \prod_{k=1}^m \frac{\lambda_{1|k-1}}{\lambda_{0|0}+k W_{1}^{-1}} \right]\\
        {\bf c} &\!\!=\! [\lambda_{1|0},\ \lambda_{0|0}]. 
    \end{align*}
Note the decoupling and monotonicity of $x_i$ and $c_i$ in both cases: $x_i$ depends only on $W_{i}$ (and not on $W_{1-i}$) and $c_i$ only on $\lambda_{1-i}$ (and not on $\lambda_i$), and they increase with $W_i$ and $\lambda_{1-i}$, respectively. 
Since throughput depends on the MHTs only via ${\bf x}$, we use ${\bf x}$ as our optimization variables to maximize the throughput. In particular, consider the following problem: 
$$\mathrm{maximize}~\{R_m({\bf x})\,|\, {\bf x} > {\bf 0}, {A}_{\max} \le {A}_{\rm thr}\},$$
where  ${A}_{\max}$ and ${A}_{\rm thr}$ are given in \eqref{eq_Abar_max} and \eqref{Age_threshold_e2e}, respectively. 
Note that the constraint ${A}_{\max} \le {A}_{\rm thr}$ is equivalent to the constraint $W_i\le \bar{W}_i$, $i = 0, 1$, where 
\begin{equation}
    \bar{W}_{0} = ({A}_{\rm thr}-{A}_0)/\gamma_{01}~\text{and}~\bar{W}_{1} = ({A}_{\rm thr}-{A}_0)/\gamma_{12} \ . 
    \label{eq_max_cutoff_bounds}
\end{equation} 
These constraints can be converted to some upper bounds $\bar{\bf x}$ on ${\bf x}$ because of the aforementioned monotonicity of $x_i$ in $W_i, i = 0, 1$. Thus, 
the problem above is equivalent to
\begin{equation} \label{eq_max_rate_strict_fidelity}
    \mathrm{maximize}~\{R_m({\bf x})\,|\, 0 \le {\bf x} \le \bar{\bf x}\}.
\end{equation}
\begin{theorem}
    The solution to this problem is given by $\mathbf{x}^* = \bar{\mathbf{x}}.$
\end{theorem}
\begin{proof}
As the main constraint is a simple bound, we will show that $R_m({\bf x})$ is component-wise monotonic and the optimal solution ${\bf x}^*$ is a corner point. In particular, partial derivatives 
\begin{align*}
    \frac{\partial R_m ({\bf x})}{\partial x_i} = \frac{c_i + (c_i-c_{1-i})x_{1-i}}{(1+x_0+x_1)^2}
\end{align*}
imply $\frac{\partial R_m ({\bf x})}{\partial x_0} > 0$ for ${\bf x} > {\bf 0}$ if $c_0 > c_1$. This means $x_0^* = \bar{x}_0$. 
On this manifold, $\frac{\partial R_m ({\bf x})}{\partial x_1} \ge 0$ iff $c_1 \ge (c_0-c_1)x_0^*$. We will show that this condition holds. First, for the multiple-memory case, starting with the Taylor series and truncating it, we obtain
\begin{align*}
\frac{c_1}{c_0-c_1} 
&= \sum_{n=1}^{\infty} \Big( \frac{c_1}{c_0} \Big)^n 
\ge  \sum_{n=1}^{K_{0}} \Big( \frac{c_1}{c_0} \Big)^n  
= \sum_{n=1}^{K_{0}} \Big( \frac{\lambda_{0|0}}{\lambda_{1|0}} \Big)^n \\
&\ge \sum_{n=1}^{K_{0}} \prod_{k=1}^n \frac{\lambda_{0|k-1}}{\lambda_{1|0} + k\bar{W}_{0}^{-1}} = x_0^*, 
\end{align*}
where the last inequality follows from $k\bar{W}_{0}^{-1} >0$  and $\lambda_{0|k-1} \le \lambda_{0|0}$ for all $k\ge 1$. 
Thus, we must have $x_1^* = \bar{x}_1$ when $c_0 > c_1$. A similar argument holds for the other case $c_0 < c_1$. As a result, we have ${\bf x}^* = \bar{\bf x}$. The single memory case can also be verified in a straightforward manner.
\end{proof}

This finding tells us that, in order to maximize the E2E entanglement throughput, for a given fidelity requirement, entangled pairs should be stored in the memories for up to $\bar{W}_i, i = 0, 1$ given in \eqref{eq_max_cutoff_bounds}. These bounds depend on the combined dephasing rates $\gamma_{01}$ and $\gamma_{12}$ of both links. Clearly, $\bar{W}_0 = \bar{W}_1$ when the memories at both end nodes have the same dephasing rate $\gamma_0$ = $\gamma_2$. In this case, we can use the same MHT for both queues. But, as we allow heterogeneous dephasing rates, MHTs are different in general. 

Finally, note that it is possible to include a constraint on the expected fidelity to the above optimization problem. In this case, the optimal solution needs not be on the boundary of the constraint set; we leave this extension for future work. 
\subsection{More General Memory Noise Model}
\label{sec_optimal_thruput_general_mem_noise}
We now return to the general memory noise model in Section~\ref{sec:Network_Model}. 
We denote the waiting times of entanglements in $\mathcal{Q}_0$ and $\mathcal{Q}_1$ by $\delta_{1L}$ and $\delta_{1R}$, respectively, with $\delta_{1L}\cdot\delta_{1R}=0$. The ages of all involved qubits are 
\begin{align*}
    a_{1L} &= \tau_{01} + \delta_{1L}, \quad a_{0R} = a_{1L}+T_{\rm swp}+\ell_{01} \\
    a_{1R} &= \tau_{12}+\delta_{1R}, \quad a_{2L} = a_{1R}+T_{\rm swp}+\ell_{12}.
\end{align*}
Recall from \eqref{e2e_fidelity} that 
$f_{e2e} \propto  v_{02}(1+2w_{02})$ with
\begin{align*}
    v_{02} &= v_0(a_{0R})v_2(a_{2L}) v_1(a_{1L})v_1(a_{1R}) v_{01}^{\rm B} v_{12}^{\rm B} \beta_1, \\
    w_{02} &= w_0(a_{0R})w_2(a_{2L}) w_1(a_{1L})w_1(a_{1R}) w_{01}^{\rm B} w_{12}^{\rm B}. \nonumber
\end{align*} 
Thus we can view $f_{e2e}$ as a function of $(\delta_{1L}, \delta_{1R})$. Since $v_i,w_i$ are monotonically decreasing in $(\delta_{1L}, \delta_{1R})$, so is $f_{e2e}$. This implies the following. 
First, the maximum fidelity is
$$ F_{\max} = f_{\rm e2e}(0,0). 
$$
Second, if $e_0$ arrives first and waits, then  $\delta_{1L} \le W_0$, $\delta_{1R}=0$, and $f_{\rm e2e}(\delta_{1L},0) \ge f_{\rm e2e}(W_0,0)$; otherwise  $f_{\rm e2e}(0,\delta_{1R}) \ge f_{\rm e2e}(0,W_1)$. Thus, for a  fixed fidelity requirement $F_{\rm req} < F_{\max}$, we can ensure $f_{\rm e2e} \ge F_{\rm req}$ by imposing upper bounds on the MHTs $W_i\le \bar{W}_i$, $i = 0, 1$ where 
$$ 
f_{\rm e2e}(0,\bar{W}_1) = f_{\rm e2e}(\bar{W}_0, 0) = F_{\rm req}.
$$
The uniqueness of $\bar{W}_i$ follows from monotonicity of $f_{\rm e2e}$. From here, the rate maximization can proceed as above. 

\section{Simulation and Evaluation}
\label{sec:evaluation}
We evaluate our models and algorithms using a discrete event simulator called SeQUeNCe \cite{wu2021sequence}. Specifically, we study the repeater capacity as a trade-off between throughput and fidelity across different memory allocations and operating modes, and demonstrate the performance gains achieved over different path lengths. The following are key parameters.



\textit{Channels and BSMs}: Quantum channels are optical fiber with  $L_0{=}32\,\text{km}$, $L_1{=}18\,\text{km}$, a peak frequency of $50$ MHz and attenuation rate of $0.2$ dB/km. Light speed in fiber is $c_0 = 2 {\cdot} 10^5\,\text{km/s}$. BSM dephasing parameters are modeled as $w_{01}^{\rm B} = e^{-L_0/2L_c}$ and $w_{12}^{\rm B} = e^{-L_1/2L_c}$ with a coherence length of $L_c=250\,\text{km}$. 
BSM depolarizing parameters due mainly to background photons are $v_{01}^{\rm B} = 0.03$ and $v_{12}^{\rm B} = 0.02$.   
We assume detectors at photonic BSMs with moderate parameters, including resolution of $100\,\text{ps}$, efficiency $\eta_{\rm d}=0.85$, and dark count rate of $50\,\text{Hz}$. Using a $100\,\text{ps}$ detection window, the effect of detector dark counts is negligible compared to our assumed background photon levels. 
Classical latencies are $\ell_{i,i+1}{=}\frac{L_i}{c_0}{+}\tau_{\rm proc}$ with processing time $\tau_{\rm proc}=10\,\mu\text{s}$. 

\textit{Memory:}  Each node has a limited number of quantum memories with an efficiency of $\eta_{\rm mem}=0.69$ and a frequency of $50$ kHz. We assume that decoherence is mainly Markovian dephasing with different rates: ${\gamma_0 = 5, \gamma_1=\gamma_2 = 10}$. The success probability of swapping is $q=0.5$. (Although this value is not critical for our simple path, it will be important in multi-repeater cases.) We use $T_{\rm swp} = 2\max(\ell_{01},\ell_{12})$. For entangled memories, we use MHTs derived in \eqref{eq_max_cutoff_bounds}.

Finally, for all cases below, we run at least $100$ independent simulations, each with a duration of $10$ seconds.


\subsection{Unit-Memory Queues}
To assess the entanglement generation rates of each link, we first simulated them independently, assuming unit memory. The mean rates are $\lambda_0^* = 76.3$ and $\lambda_1^* = 244.5$ entangled pairs per second (epps). These values are then used for our models to study the E2E entanglement throughput. 
Fig.~\ref{fig_unit_memory} plots the maximum achievable throughput and the average E2E fidelity as a function of the required fidelity with $F_{\max}=0.9033$, including the case with unit-efficiency detectors and swapping to demonstrate the performance limit. 

Both our analysis and simulation results confirm the swapping capacity trade-off: ensuring higher minimum fidelity $F_{\text{req}}$ demands smaller MHTs, resulting in lower throughput due to frequent resets and discards. Conversely, lower required fidelity permits larger MHTs, which significantly increases the throughput. Our renewal model accurately tracks simulation results for both throughput and average fidelity across various $F_{\text{req}}$, validating our assumptions and modeling approach in this case. CTMC and BDP models can also be used but less accurate; CTMC tends to underestimate the rate due to its  exponential‑distribution approximation, while BDP overestimates CTMC by assuming negligible reset and swap delays. The deviation between BDP and CTMC models is noticeable under unit-efficiency detectors due to such simplifying assumptions. Specifically, in this case, the interarrival times of the queues become more comparable with those delays.

\begin{figure}[t]
    \centering
    \includegraphics[width=1\linewidth]{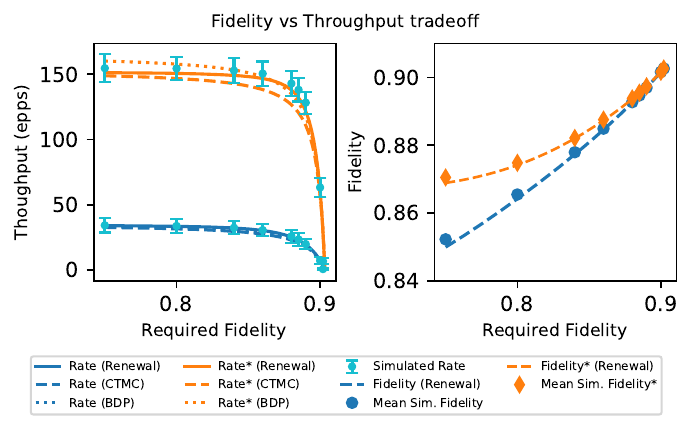}
    \caption{\small Comparison of models and simulation in unit-memory case. \texttt{Rate*} and \texttt{Fidelity*} correspond to near ideal setting with parameters: $\eta_{\rm mem}=0.9, \eta_{\rm d}=1$ and deterministic swapping $q=1$. Error bars denote means and standard deviations of simulation results.}
    \label{fig_unit_memory}
    \vspace{-3mm}
\end{figure}

\subsection{Multiple-Memory Queues}
Suppose the repeater has 6 memories in total, which can be assigned to the links in different ways. 
We assume \textit{memory multiplexing}, where multiple pairs can be used for parallel (or time-division multiplexed) entanglement generation. 
This yields state-dependent arrival rates $\lambda_{i|N_i} = (K_i-N_i)\lambda_i^*$ in our queue models. Fig.~\ref{fig_multi_memory} illustrates the swapping capacity via our models and simulation for various memory allocations.

\begin{figure*}[t]
    \centering
    \includegraphics[width=1\textwidth]{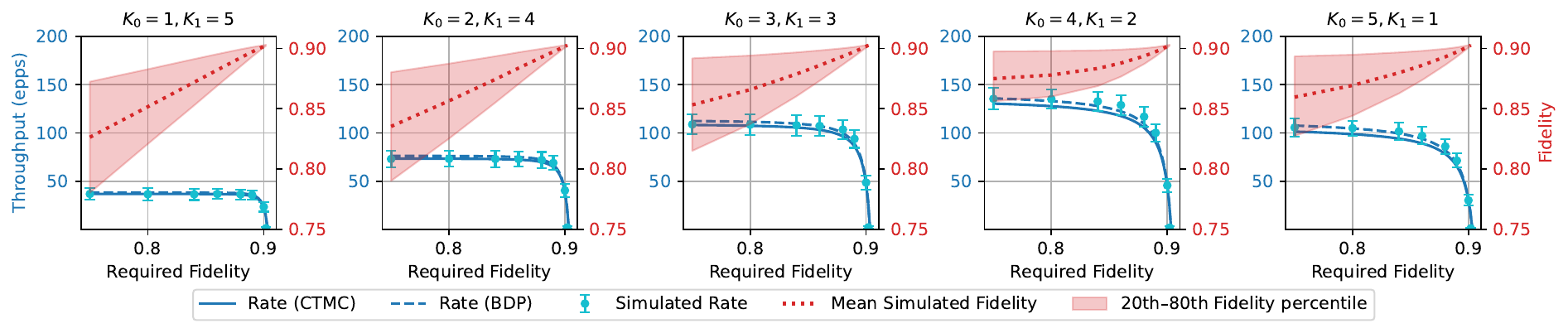}
    \caption{\small Throughput and fidelity when varying memory allocation and multiplexing.}
    \label{fig_multi_memory}
    \vspace{-3mm}
\end{figure*}

We note the following. First, in all settings, both CTMC and BDP models closely match the simulated throughput and show small deviations in some cases. 
Second, the impact of multiplexing is clearly demonstrated. By enabling parallelism, multiplexing substantially increases achievable throughput. Crucially, the capacity trade-off between throughput and fidelity extends to a higher range of required fidelity. For instance, when $K_0=1$ and $K_1=5$, the maximum throughput is relatively preserved up to a minimum fidelity of $0.89$. This is because $\mathcal{Q}_0$ is the bottleneck in this case due to its much lower generation rate and queue size compared to $\mathcal{Q}_1$.  
Shifting memories from $\mathcal{Q}_1$ to $\mathcal{Q}_0$ helps alleviate this bottleneck: not only does it increase the throughput as more pairs are matched, it also reduces the waiting time of entanglements at $\mathcal{Q}_1$, thereby improving the E2E fidelity. 
Importantly, our model accurately captures this shift, proving its capability to account for both resource allocation and parallel entanglement generation. This allows for easy identification of the optimal allocation, i.e., $K_0=4$ and $K_1 =2$.




The strong agreement between our models and simulations across various memory allocations and operating modes proves that our model captures the key system dynamics and can reliably predict performance under diverse configurations.

\subsection{Impact of Model-Based Selection of MHTs}
First, we demonstrate the optimality of our model-based MHTs $\bar{W}_i$ given by~\eqref{eq_max_cutoff_bounds}. To this end, we consider the case with $K_0=4, K_1=2$, and $F_{\rm req}=0.88$, and manually vary $W_i$ around $\bar{W}_i$. Fig.~\ref{fig_heatmap} shows the heatmap of the average simulated throughput, which clearly proves the optimality of $\bar{W}_i$. 

\begin{figure}
    \centering
    \includegraphics[width=0.9\linewidth]{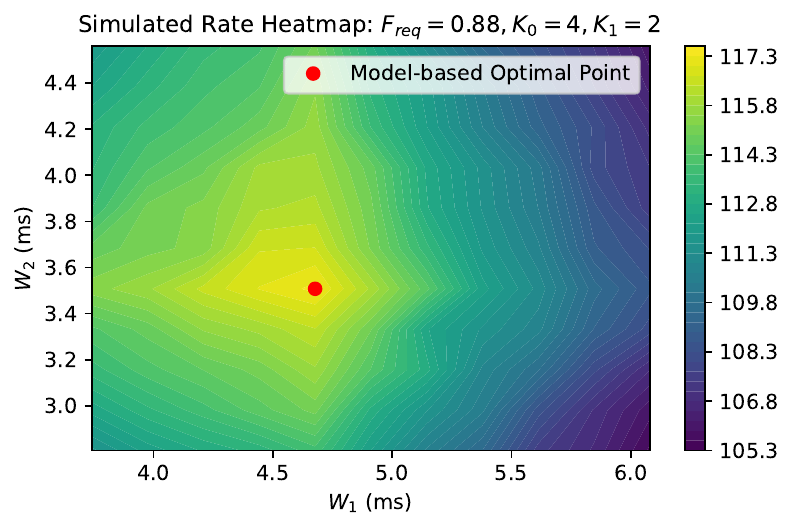}
    \vspace{-2mm}
    \caption{\small Average throughput when varying MHTs $W_i$.}
    \label{fig_heatmap}
    \vspace{-2mm}
\end{figure}

Next, we show that without our model-based MHTs, such as when MHTs are set to the maximum extent of the memory coherence time $1/\gamma_i$, the swapping capacity can reduce significantly. To this end, we consider the memory allocation $K_0=4, K_1=2$ and double the distances from $(L_0, L_1)=(32,18)$ to $(2L_0, 2L_1)=(64,36)$; in the latter case we keep the depolarizing parameters $v_{ij}^{\rm B}$ and the coherence length $L_c$ constant but allow dephasing parameters $w_{ij}^{\rm B}$ and $T_{\rm swp}$ to decay with the channel lengths, leading to $F_{\max}=0.8571$. The throughput comparison is given in Fig.~\ref{fig:eval_ttl_improvement}; dotted lines represent the rates without enforcing MHTs. Clearly, without MHTs, many generated entanglements fail the requirement, especially for higher $F_{\rm req}$ and longer path due to increased decoherence; e.g., the rate decreases by 1-2 orders of magnitude around $F_{\rm req}=0.88$ for the shorter path and $0.84$ for the longer one. Our model-based MHTs retain the usable entanglement rate, particularly at stringent fidelity requirements.

\begin{figure}[t]
    \centering
    \includegraphics[width=0.825\linewidth]{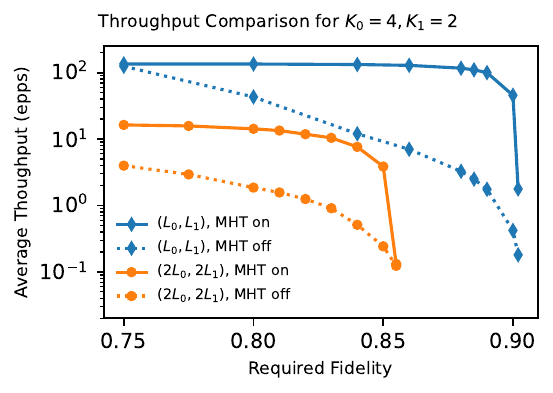}
    \vspace{-2mm}
    \caption{\small Comparing throughput with and without model-based MHTs.}
    \label{fig:eval_ttl_improvement}
    \vspace{-3mm}
\end{figure}

\section{Concluding Remarks}
In this paper, we investigated the fundamental swapping capacity trade-off between the achievable E2E entanglement throughput and the minimum fidelity of a quantum repeater with multi-memory links. We developed a comprehensive noise model incorporating both time-dependent memory decoherence and link-level errors. By leveraging renewal reward and queueing theories, we derived practical throughput models that can be used to find optimal cutoff times of quantum memories that maximize E2E throughput while guaranteeing a minimum required fidelity. Our extensive simulations validated the accuracy of our models, which closely tracked the throughput and average fidelity across various operating conditions. We showed that failing to enforce optimal memory cutoff times, especially in longer paths with increased decoherence, leads to a significant reduction in usable E2E entanglement rate. Furthermore, our proposed framework effectively captures the benefits of resource allocation, allowing for efficiently identifying optimal memory assignments to maximize network performance. This work provides a crucial theoretical and practical tool for the design and operation of quantum repeaters, enabling precise management of the rate-fidelity trade-off and realizing the potential of quantum memory networks for distributed quantum applications. 

Future work will focus on extending our analysis to more complex quantum network architectures. This includes modeling the E2E performance over multi-repeater paths, which introduces cascaded noise effects and memory management challenges. Furthermore, the framework will need to support arbitrary network topologies, necessitating the integration of dynamic routing and entanglement swapping protocols with advanced techniques like entanglement purification to maintain fidelity across the full quantum internet infrastructure.

\newpage
\bibliographystyle{IEEEtran}
\bibliography{ref}

\end{document}